\documentclass{article}
\usepackage{ijcai21}

\usepackage{times}
\usepackage{soul}
\usepackage{url}
\usepackage[hidelinks]{hyperref}
\usepackage[utf8]{inputenc}
\usepackage[small]{caption}
\usepackage{graphicx}
\usepackage{amsmath}
\usepackage{amsthm}
\usepackage{booktabs}
\newtheorem{theorem}{Theorem}

\usepackage[style=authoryear,maxbibnames=99,dashed=false]{biblatex}
\bibliography{ijcai21}
\providecommand{\citet}[2][]{\ifx&#1&\textcite{#2}\else\textcite[#1]{#2}\fi}
\renewcommand{\cite}[1]{\parencite{#1}}

\title{Two-Stage Facility Location Games with Strategic Clients and Facilities}

\author{
Simon Krogmann$^1$
\and
Pascal Lenzner$^1$\and
Louise Molitor$^1$\And
Alexander Skopalik$^2$
\affiliations
$^1$Hasso Plattner Institute, University of Potsdam\\
$^2$Mathematics of Operations Research, University of Twente\\
\emails
\{simon.krogmann, pascal.lenzner, louise.molitor\}@hpi.de,
a.skopalik@utwente.nl
}

\newtheorem{lemma}{Lemma}
\newtheorem{definition}{Definition}
\newtheorem{corollary}{Corollary}

\usepackage{amsfonts}
\usepackage
[
    ruled,         
    vlined,        
    linesnumbered, 
]
{algorithm2e} 
\usepackage{etoolbox}
\usepackage{amssymb}

\makeatletter
\patchcmd{\@algocf@start}
  {-1.5em}
  {0pt}
  {}{}
\makeatother
\usepackage{tikz}
\usetikzlibrary{decorations.pathreplacing,calligraphy,arrows.meta}
\usepackage
[
    noabbrev,   
    nameinlink, 
]
{cleveref}

\usepackage{todonotes}

\newcommand{\limtmpmodel}{$2$-FLG} 

\DeclareMathOperator*{\argmin}{arg\,min}
\newcommand{\ensurenewline}{%
\ifhmode
    \newline
\fi
}

\newcommand*{\trheight}{1.03923}
\newcommand*{\trlength}{1.2}

\newcommand{\s}{\mathbf{s}}

\begin{document}

\maketitle

\begin{abstract}
    We consider non-cooperative facility location games where both facilities and clients act strategically and heavily influence each other. This contrasts established game-theoretic facility location models with non-strategic clients that simply select the closest opened facility.
    In our model, every facility location has a set of attracted clients and each client has a set of shopping locations and a weight that corresponds to her spending capacity. Facility agents selfishly select a location for opening their facility to maximize the attracted total spending capacity, whereas clients strategically decide how to distribute their spending capacity among the opened facilities in their shopping range. We focus on a natural client behavior similar to classical load balancing: our selfish clients aim for a distribution that minimizes their maximum waiting times for getting serviced, where a facility's waiting time corresponds to its total attracted client weight.

    We show that subgame perfect equilibria exist and give almost tight constant bounds on the Price of Anarchy and the Price of Stability, which even hold for a broader class of games with arbitrary client behavior. Since facilities and clients influence each other, it is crucial for the facilities to anticipate the selfish clients' behavior when selecting their location. For this, we provide an efficient algorithm that also implies an efficient check for equilibrium.
    Finally, we show that computing a socially optimal facility placement is NP-hard and that this result holds for all feasible client weight distributions.
\end{abstract}

\section{Introduction}
Facility location problems are widely studied in Operations Research, Economics, Mathematics, Theoretical Computer Science, and Artificial Intelligence. In essence, in these problems facilities must be placed in some underlying space to serve a set of clients that also live in that space. Famous applications of this are the placement of hospitals in rural areas to minimize the emergency response time or the deployment of wireless Internet access points to maximize the offered bandwidth to users. These problems are purely combinatorial optimization problems and can be solved via a rich set of methods.
Much more intricate are facility location problems that involve competition, i.e., if the facilities compete for the clients. These settings can no longer be solved via combinatorial optimization and instead, methods from Game Theory are used for modeling and analyzing them.

The first model on competitive facility location was the famous \emph{Hotelling-Downs model}, first introduced by \citet{hotelling} and later refined by \citet{downs}.
Their original interpretations are selling a commodity in the main street of a town, and parties placing themselves in a political left-to-right spectrum, respectively.
They assume a one-dimensional market on which clients are uniformly distributed and where $k$ facility agents each want to place a single facility on the market. Each facility gets the clients, to which their facility is closest.
\citet{voronoigames} introduced Voronoi games on networks, that move the problem onto a graph and assume discrete clients on each node.

The models mentioned above are one-sided, i.e., only the facility agents face a strategic choice while the clients simply patronize their closest facility independently of the choices of other clients. Obviously, realistic client behavior can be more complex than this.
For example, a client might choose not to patronize any facility, if there is no facility sufficiently close to her. This setting was recently studied by \citet{feldman-hotelling}, \citet{hotelling-limited-attraction} and \citet{hotelling-line-bubble} albeit with continuous clients on a line.
In their model with limited attraction ranges, clients split their spending capacity uniformly among all facilities that are within a certain distance. In contrast to the Hotelling-Downs model, pure Nash equilibria always exist. In another related variant by \citet{fournier2020spatial}, clients that have multiple facilities in their range choose the nearest facilities.
Another natural client behavior is that they might avoid crowded facilities to reduce waiting times.
This notion was introduced to the Hotelling-Downs model by \citet{load-balancing}, also on a line.
Clients consider a linear combination of both distance and waiting time, as they want to minimize the total time spent visiting a facility.
This models clients that perform load balancing between different facilities.
\citet{Peters2018} prove the existence of subgame perfect equilibria for certain trade-offs of distance and waiting time for two, four and six facilities and they conjecture that equilibria exist for all cases with an even number of facilities for client utility functions that are heavily tilted towards minimizing waiting times. \citet{hotelling-load-balancing} investigated the existence of approximate pure subgame perfect equilibria for Kohlberg's model and their results indicate that $1.08$-approximate equilibria exist.
The most notable aspect of Kohlberg's model is that it is two-sided, i.e., both facility and client agents act strategically. This implies that the facility agents have to anticipate the client behavior, in particular the client equilibrium. For Kohlberg's model \citet{hotelling-load-balancing} show that this entails the highly non-trivial problem of solving a complex system of equations.

In this paper we present a very general two-sided competitive facility location model that is essentially a combination of the models discussed above. Our model has an underlying host graph with discrete weighted clients on each vertex. The host graph is directed, which allows to model limited attraction ranges, and we have facilities and clients that both face strategic decisions. Most notably, in contrast to Kohlberg's model and despite our model's generality, we provide an efficient algorithm for computing the facilities' loads in a client equilibrium. Hence, facility agents can efficiently anticipate the client behavior and check if a game state is in equilibrium.

\subsection{Further Related Work}
Voronoi games were introduced by \citet{voronoi1d} on a line. For the version on networks by \citet{voronoigames}, the authors show that equilibria may not exist and that existence is NP-hard to decide.
Also, they investigate the ratio between the social cost of the best and the worst equilibrium state, where the social cost is measured by the total distance of all clients to their selected facilities. With $n$ the number of clients and $k$ the number of facilities, they prove bounds of $\Omega(\sqrt{n/k})$ and $\mathcal{O}(\sqrt{kn})$.
While we are not aware of other results on general graphs, there is work for specific graph classes:
\citet{voronoi-cycle} limit their investigation to cycle graphs and characterize the existence of equilibria and bound the Price of Anarchy (PoA)~\cite{poa} and the Price of Stability (PoS)~\cite{pos} to $\frac94$ and $1$, respectively.
Additionally, there are many closely related variants with two agents:
restaurant location games \cite{restaurant-location}, a variant by~\citet{duopoly-voronoi}, and a multi round version~\cite{voronoi-multi-round}.
Moreover, there are variants played in $k$-dimensional space: \citet{voronoi-kd-voters}, \citet{voronoi1d}, \citet{voronoi-choice}.
To the best of our knowledge, there is no variant with strategic clients aiming at minimizing their maximum waiting time.

A concept related to our model are utility systems, as introduced by \citet{vetta-utility-system}.
Agents gain utility by selecting a set of acts, which they choose from a collection of subsets of a groundset.
Utility is assigned by a function that takes the selected acts of all agents as an input.
Two special types are considered: basic and valid utility systems.
For the former, it is shown that pure Nash equilibria (NE) exist.
For the latter, no NE existence is shown but the PoA is upper bounded by 2.
We show in \Cref{sec:utility-systems} that our model with load balancing clients is a valid but not a basic utility system.

Covering games~\cite{gairing-covering-games} correspond to a one-sided version of our model, i.e., where clients simply distribute their weight uniformly among all facilities in their shopping range. There, pure NE exist and the PoA is upper bounded by~$2$.
More general versions are investigated by \citet{3g-market-sharing} and \citet{alex-market-sharing} in the form of market sharing games. In these models, $k$ agents choose to serve a subset of $n$ markets. Each market then equally distributes its utility among all agents who serve it.
\citet{alex-market-sharing} show a PoA of $2-\frac1k$ for their game.

Recently
\citet{alex-network-investment} introduced a model which considers an inherent load balancing problem, however, each facility agent can create and choose multiple facilities and each client agent chooses multiple facilities.

For further related models we refer to the excellent surveys by \citet{ELT93} and \citet{RE05}.

\subsection{Model and Preliminaries}
We consider a game-theoretic model for non-cooperative facility location, called the \emph{Two-Sided Facility Location Game (\limtmpmodel{})}, where two types of agents, $k$ \emph{facilities} and $n$ \emph{clients}, strategically interact on a given vertex-weighted directed host graph $H = (V,E,w)$, with $V = \{v_1,\dots,v_n\}$, where $w:V \to \mathbb{N}$ denotes the vertex weight. Every vertex $v_i \in V$ corresponds to a client with weight $w(v_i)$, that can be understood as her spending capacity, and at the same time each vertex is a possible location for setting up a facility for any of the $k$ facility agents $\mathcal{F} = \{f_1,\dots,f_k\}$. Any client $v_i \in V$ considers visiting a facility in her \emph{shopping range} $N(v_i)$, i.e., her direct closed neighborhood $$N(v_i) = \{v_i\} \cup \{z \mid (v_i,z) \in E\}.$$ Moreover, let $$w(X) = \sum_{v_i \in X}w(v_i),$$ for any $X \subseteq V$, denote the total spending capacity of the client subset $X$.

In our setting the strategic behavior of the facility and the client agents influences each other. Facility agents select a location to attract as much client weight as possible, whereas clients strategically decide how to distribute their spending capacity among the facilities in their shopping range.
More precisely, each facility agent $f_j \in \mathcal{F}$ selects a single location vertex $s_j \in V$ for setting up her facility, i.e., the strategy space of any facility agent $f_j\in \mathcal{F}$ is $V$. Let $\mathbf{s} = (s_1,\dots,s_k)$ denote the \emph{facility placement profile}. And let $\mathcal{S} = V^k$ denote the set of all possible facility placement profiles.
We will sometimes use the notation $\mathbf{s} = (s_j,s_{-j})$, where $s_{-j}$ is the vector of strategies of all facilities agents except $f_j$.
Given $\mathbf{s}$, we define the \emph{attraction range} for a facility $f_j$ on location $s_j \in V$ as $$A_\mathbf{s}(f_j) = \{s_j\} \cup \{v_i \mid (v_i,s_j) \in E\}.$$ We extend this to sets of facilities $F \subseteq \mathcal{F}$ in the natural way, i.e., $$A_\mathbf{s}(F) = \{s_j \mid f_j \in F\} \cup \{v_i \mid (v_i,s_j) \in E, f_j\in F\}.$$ Moreover, let $$w_\mathbf{s}(\mathcal{F}) = \sum_{v_i \in A_\mathbf{s}(\mathcal{F})} w(v_i).$$

We assume that all facilities provide the same service for the same price and arbitrarily many facilities may be co-located on the same location.
Each client $v_i\in V$ strategically decides how to distribute her spending capacity $w(v_i)$ among the opened facilities in her shopping range~$N(v_i)$. For this, let $$N_{\mathbf{s}}(v_i) = \{f_j \mid s_j \in N(v_i)\}$$ denote the set of facilities in the shopping range of client $v_i$ under $\mathbf{s}$.

Let $\sigma: \mathcal{S} \times V \to \mathbb{R}_+^k$ denote the \emph{client weight distribution function}, where $\sigma(\mathbf{s},v_i)$ is the weight distribution of client $v_i$ and $\sigma(\mathbf{s},v_i)_{j}$ is the weight distributed by $v_i$ to facility~$f_j$. We say that $\sigma$ is \emph{feasible} for $\mathbf{s}$, if all clients having at least one facility within their shopping range distribute all their weight to the respective facilities and all other clients distribute nothing.
Formally, $\sigma$ is feasible for $\mathbf{s}$, if for all $v_i\in V$ we have $\sum_{f_j \in N_\mathbf{s}(v_i)} \sigma(\mathbf{s},v_i)_j = w(v_i)$, if $N_\mathbf{s}(v_i) \neq \emptyset$, and $\sigma(\mathbf{s},v_i)_j = 0$, for all $1\leq j \leq k$, if $N_\mathbf{s}(v_i) = \emptyset$.
We use the notation $\sigma = (\sigma_i,\sigma_{-i})$ and $(\sigma_i',\sigma_{-i})$ denotes the changed client weight distribution function that is identical to $\sigma$ except for client $v_i$, which plays $\sigma'(\mathbf{s},v_i)$ instead of $\sigma(\mathbf{s},v_i)$.

Any state $(\mathbf{s},\sigma)$ of the \limtmpmodel{} is determined by a facility placement profile~$\mathbf{s}$ and a feasible client weight distribution function $\sigma$.
A state $(\mathbf{s},\sigma)$ then yields a \emph{facility load} $\ell_j(\mathbf{s},\sigma)$ with $\ell_j(\mathbf{s},\sigma) = \sum_{i=1}^n \sigma(\mathbf{s},v_i)_j$ for facility agent $f_j$. Hence, $\ell_j(\mathbf{s},\sigma)$ naturally models the total congestion for the service offered by the facility of agent $f_j$ induced by~$\sigma$. A facility agent $f_j$ strategically selects a location $s_j$ to maximize her induced facility load $\ell_j(\mathbf{s},\sigma)$. We assume that the service quality of facilities, e.g., the waiting time, deteriorates with increasing congestion. Hence, for a client the facility load corresponds to the waiting time at the respective facility.

There are many ways of how clients could distribute their spending capacity.
As a proof-of-concept we consider the \emph{load balancing \limtmpmodel{}} with
\emph{load balancing clients}, i.e., a natural strategic behavior where client $v_i$ strategically selects $\sigma(\mathbf{s},v_i)$ to minimize her maximum waiting time. More precisely, client $v_i$ tries to minimize her \emph{incurred maximum facility load} over all her patronized facilities (if any). More formally, let $$P_i(\mathbf{s},\sigma) = \{j \mid \sigma(\mathbf{s},v_i)_j > 0\}$$ denote the set of facilities patronized by client $v_i$ in state $(\mathbf{s},\sigma)$. Then client $v_i$'s incurred maximum facility load in state $(\mathbf{s},\sigma)$ is defined as
$$L_i(\mathbf{s},\sigma) = \max_{j \in P_i(\mathbf{s},\sigma)} \ell_j(\mathbf{s},\sigma).$$
We say that $\sigma^*$ is a \emph{client equilibrium weight distribution}, or simply a \emph{client equilibrium}, if for all $v_i \in V$ and for all placement profiles $\mathbf{s}$ we have that $L_i(\mathbf{s},(\sigma_{i}^*,\sigma_{-i})) \leq L_i(\mathbf{s},(\sigma_{i}',\sigma_{-i}))$ for all possible weight distributions $\sigma'(\mathbf{s},v_i)$ of client $v_i$. See \Cref{fig:figure1} for an illustration of the load balancing \limtmpmodel{}.
\begin{figure}[t]
    \centering
    \includegraphics[width=8.0cm]{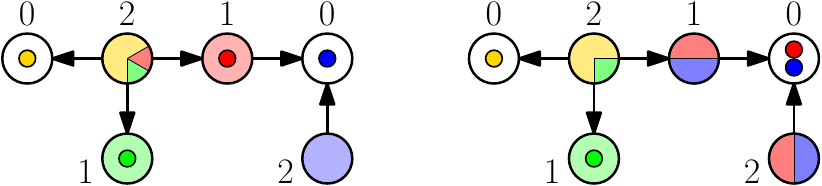}
    \caption{
        Example of the load balancing \limtmpmodel{}. The clients (vertices) split their weight (shown by numbers) among the facilities (colored dots) in their shopping range. The client distributions are shown by colored pie charts. Left: The blue facility receives a load of $2$ while all other facilities get a load of $\frac43$. The left client with weight $2$ distributes weight $\frac{4}{3}$ to the yellow facility and $\frac{1}{3}$ to both the green and the red facility.
        The state is not in SPE as the red facility can improve her load to $\frac32$ by co-locating with the blue facility. Right: A SPE for this instance, all facilities have a load of $\frac32$.}
    \label{fig:figure1}
\end{figure}

We define the \emph{stable states} of the \limtmpmodel{} as \emph{subgame perfect equilibria (SPE)}, since we inherently have a two-stage game. First, the facility agents select locations for their facilities and then, given this facility placement, the clients strategically distribute their spending capacity among the facilities in their shopping range.
A state $(\mathbf{s},\sigma)$ is in SPE, or \emph{stable}, if
\begin{itemize}
    \item[(1)] $\forall f_j \in \mathcal{F},  \forall s_j' \in V$: $\ell_j(\mathbf{s},\sigma) \geq \ell_j((s_j',s_{-j}),\sigma)$ and
    \item[(2)] $\forall \mathbf{s}' \in \mathcal{S}, \forall v_i \in V: L_i(\mathbf{s}',\sigma) \leq L_i(\mathbf{s}',(\sigma_{i}',\sigma_{-i}))$ for all feasible client weight distribution functions $\sigma'$.
\end{itemize}

We say that client $v_i$ is \emph{covered by $\mathbf{s}$}, if $N_\mathbf{s}(v_i) \neq \emptyset$, and \emph{uncovered by $\mathbf{s}$}, otherwise. Let $$C(\mathbf{s}) = \{v_i \mid v_i \in V, N_\mathbf{s}(v_i) \neq \emptyset\}$$ denote the set of covered clients under facility placement~$\mathbf{s}$.
We will compare states of the \limtmpmodel{} by measuring their \emph{social welfare} that is defined as the \emph{weighted participation rate} $$W(\mathbf{s}) = w(C(\mathbf{s})) = \sum_{v_i \in C(\mathbf{s})}w(v_i),$$ i.e., the total spending capacity of all covered clients.
For a host graph $H$ and a number of facility agents $k$, let $\text{OPT}(H,k)$ denote the facility placement profile that maximizes the weighted participation rate $W(\text{OPT}(H,k))$ among all facility placement profiles with $k$ facilities on host graph~$H$.

We measure the inefficiency due to the selfishness of the agents via the Price of Anarchy (PoA) and the Price of Stability (PoS). Let $\text{bestSPE}(H,k)$ (resp. $\text{worstSPE}(H,k)$) denote the SPE with the highest (resp. lowest) social welfare among all SPEs for a given host graph $H$ and a facility number $k$. Moreover, let $\mathcal{H}$ be the set of all possible host graphs $H$.
Then the PoA is defined as $$\text{PoA} :=\max_{H\in \mathcal{H},k} 
\frac{
W(\text{OPT}(H,k))}{W(\text{worstSPE}(H,k))},$$

whereas the PoS is defined as
$$\text{PoS}:=\max_{H\in \mathcal{H},k} \frac{W(\text{OPT}(H,k))}{W(\text{bestSPE}(H,k))}.$$

We study dynamic properties of the \limtmpmodel{}.
Let an \emph{improving move} by some (facility or client) agent be a strategy change that improves the agent's utility.
A game has the \emph{finite improvement property (FIP)} if all sequences of improving moves are finite.
The FIP is equivalent to the existence of an \emph{ordinal potential function}~\cite{MS96}.

\subsection{Our Contribution}
We introduce and analyze the \limtmpmodel{}, a general model for competitive facility location games, where facility agents and also client agents act strategically.
We focus on the load balancing \limtmpmodel{}, where clients selfishly try to minimize their maximum waiting times that not only depend on the placement of the facilities but also on the behavior of all other client agents.
We show that client equilibria always exist and that all client equilibria are equivalent from the facility agents' point-of-view.
Additionally, we provide an efficient algorithm for computing the facility loads in a client equilibrium that enables facility agents to efficiently anticipate the clients' behavior. This is crucial in a two-stage game-theoretic setting. Moreover, since there are only $n$ possible locations for facilities, we can efficiently check if a given state of the load balancing \limtmpmodel{} is in SPE. Using a potential function argument, we can show that a SPE always exists.

Finally, we consider the \limtmpmodel{} with an arbitrary feasible client weight distribution function. For this broad class of games, we prove that the PoA is upper bounded by $2$ and we give an almost tight lower bound of $2 - \frac{1}{k}$ on the PoA and~PoS. This implies an almost tight PoA lower bound for the load balancing \limtmpmodel{}. Furthermore, we show that computing a social optimum state for the \limtmpmodel{} with an arbitrary feasible client weight distribution function $\sigma$ is NP-hard for all feasible $\sigma$, hence, also for the load balancing \limtmpmodel{}.

\section{Load Balancing Clients}
In this section we analyze the load balancing \limtmpmodel{} in which we consider not only strategic facilities that try to get patronized by as many clients as possible but we also have selfish clients that strategically distribute their spending capacity to minimize their maximum waiting time for getting serviced.
We start with a crucial statement that enables the facility agents to anticipate the clients' behavior.
\begin{theorem}
\label{theo:existence}
For a facility placement profile $\mathbf{s}$, a client equilibrium $\sigma$ exists and every client equilibrium induces the same facility loads $(\ell_1(\mathbf{s},\sigma),\ldots,\ell_k(\mathbf{s},\sigma))$.
\end{theorem}
\begin{proof}
    We consider the following optimization problem (EQ):
    \begin{align*}
        &\min_{\sigma} \sum_{i=1}^k \ell_i(\mathbf{s},\sigma)^2\\
        \text{subject to}\\
        \sigma(\mathbf{s},v_i)_j &\ge 0 &\text{ for all } f_j \in N_\mathbf{s}(v_i)\\
        \sigma(\mathbf{s},v_i)_j &= 0   &\text{ for all } f_j \notin N_\mathbf{s}(v_i)\\
        \sum_{f_j \in N_\mathbf{s}(v_i)} \sigma(\mathbf{s},v_i)_j &= w(v_i)  &\text{ if }  N_\mathbf{s}(v_i) \ne \emptyset\\
    \end{align*}
    It is easy to see that an optimal solution
    $\sigma$ of EQ is a client equilibrium. For the sake of contradiction, assume that there exists a client $v_i$ and two facility agents $f_p$ and $f_q$ with $\ell_q(\mathbf{s},\sigma) > \ell_p(\mathbf{s},\sigma)$ and $\sigma(\mathbf{s},v_i)_q > 0$.
    However, this contradicts the optimality of $\sigma$ as the KKT conditions~\cite{KKT} demand that $\ell_q(\mathbf{s},\sigma) \le \ell_p(\mathbf{s},\sigma)$ for all $f_p,f_q \in N_\mathbf{s}(v_i)$ with $\sigma(\mathbf{s},v_i)_q > 0$. Moreover, the KKT conditions are precisely the conditions of a client equilibrium, hence every equilibrium is an optimal solution of EQ.

    Observe that the objective of EQ is convex in the facilities' loads $\ell_1(\mathbf{s},\sigma),\ldots,\ell_k(\mathbf{s},\sigma)$ and the set of feasible solutions is compact and convex.
    Suppose there are two global optima $\sigma$ and $\sigma'$ of EQ. By convexity of the objective function, we must have $\ell_j(\mathbf{s},\sigma) = \ell_j(\mathbf{s},\sigma')$ for all facility agents $f_j$ as otherwise a convex combination of $\sigma$ and $\sigma'$ would yield a feasible solution for EQ with smaller objective function value.
\end{proof}

\noindent Two facility agents sharing a client have equal load if the shared client puts weight on both of them:

\begin{lemma}
    \label{lemma:shared-client-equal-load-balancing}
    In the load balancing \limtmpmodel{}, for a facility placement $\mathbf{s}$ let $\sigma$ be a client equilibrium. If there are two facility agents $f_p$ and $f_q$ and a client $v_i$ with $p, q \in P_i(\s, \sigma)$, then $\ell_p(\s,\sigma)=\ell_q(\s,\sigma)$.
\end{lemma}

\begin{proof}
    Let $v_i$ be a client and $p$ be the agent with the highest load in $P_i(\s, \sigma)$.
    Assume that there is an agent $q \in P_i(\s, \sigma)$ with $\ell_p(\s,\sigma)>\ell_q(\s,\sigma)$.
    In this case, the client $v_i$ decreases her weight on $f_p$ (and all facility agents in $P_i(\s, \sigma)$ with the same load) and increases her weight on~$f_q$, decreasing her total costs.
    This contradicts $\sigma$ being a client equilibrium.
\end{proof}

\noindent Next, we define a \emph{shared client set}, which represents a set of facility agents who share weight of the same clients.

\begin{definition}
    \label{def:shared-client-set}
    For a facility placement profile $\mathbf{s}$, let $f_p$ be an agent, $\sigma$ be a client equilibrium.
    We define a shared client set of facility agents $S_\sigma(f_p)$, such that
    \begin{itemize}
        \item[(1)] $f_p \in S_{\sigma}(f_p)$ and
        \item[(2)] for two facility agents $f_q, f_r$: If $f_q\in S_\sigma(f_p)$ and there is a client $v_i$ with $q,r \in P_i(\s, \sigma)$, then $f_r\in S_{\sigma}(f_p)$.
    \end{itemize}
\end{definition}

\noindent We prove two properties of such a shared client set:
First, all facility agents in a shared client set have the same load, and second, a client's weight is either completely inside or completely outside a shared client set in a client equilibrium.

\begin{lemma}
    \label{lemma:shared-set-equal-load}
    For a facility placement $\mathbf{s}$
    in a client equilibrium $\sigma$, for every $f_q, f_r \in S_{\sigma}(f_p)$
    we have
    $\ell_q(\s,\sigma) = \ell_r(\s,\sigma)$.
\end{lemma}

\begin{proof}
    As $f_q$ and $f_r$ are both members of $S_{\sigma}(f_p)$ there exists a sequence of facility agents $F = (f_q, f_{i_1}, f_{i_2}, \dots, f_r)$, in which two adjacent facility agents share a client.
    By \Cref{lemma:shared-client-equal-load-balancing}, each pair of neighbors in $F$ has identical loads.
    Thus, $\ell_q(\s,\sigma) = \ell_r(\s,\sigma)$.
\end{proof}

\noindent The next lemma follows from \Cref{def:shared-client-set}:

\begin{lemma}
    \label{lemma:shared-set-whole-client}
    For a facility placement $\mathbf{s}$,
    in a client equilibrium $\sigma$ for every client $v_i$ and facility agent $f_p$ with $p \in P_i(\s, \sigma)$, we have that for every facility agent $f_r \notin S_{\sigma}(f_p)$ it holds that $r \notin P_i(\s, \sigma)$.
\end{lemma}

\noindent Additionally, we show that each facility agent's load can only take a limited number of values.

\begin{lemma}
    \label{lemma:limited-values-load-balancing}
    For a facility placement profile $\mathbf{s}$,
    in a client equilibrium $\sigma$ a facility agent's load can only take a value of the form $\frac{x}{y}$ for $x \leq w_\mathbf{s}(\mathcal{F})$ and $y \leq k$ with $x, y \in \mathbb{N}$.
\end{lemma}

\begin{proof}
    If a client is shared between two facilities, these two facilities must, by \Cref{lemma:shared-client-equal-load-balancing}, have the same load.
    We consider an arbitrary facility agent $f_j$ and her shared client set $S_{\sigma}(f_j)$.
    All facility agents in $S_{\sigma}(f_j)$ have the same load by \Cref{lemma:shared-set-equal-load} and all clients which have weight on a facility agent on $S_{\sigma}(f_j)$ have their complete weight inside $S_{\sigma}(f_j)$ by \Cref{lemma:shared-set-whole-client}.
    Therefore, the sum of loads of the facility agents $S_{\sigma}(f_j)$ must be an integer $i \leq w_\mathbf{s}(\mathcal{F})$.
    Thus, the load of $f_j$ is $\frac{i}{|S_{\sigma}(f_j)|}$.
    Since $i \leq w_\mathbf{s}(\mathcal{F})$ (sum of client weights) and $|S_\sigma(f_j)| \leq k$ (number of facility agents) with $i, |S_\sigma(f_j)| \in \mathbb{N}$, the lemma is true.
\end{proof}

\begin{definition}
    For a facility placement profile $\mathbf{s}$, a set of facility agents $\emptyset \subset M \subseteq \mathcal{F}$ is a \emph{minimum neighborhood set} (MNS) if for all
    $$\emptyset \subset T \subseteq \mathcal{F}\text:~\frac{w(A_\mathbf{s}(M))}{|M|} \leq \frac{w(A_\mathbf{s}(T))}{|T|}\text.$$
    We define the \emph{minimum neighborhood ratio} (MNR) as
    $$
    \rho_\mathbf{s} := \frac{w(A_\mathbf{s}(M))}{|M|}\text,
    $$
    with $M$ being a \emph{MNS}.
\end{definition}

\noindent We show that a \emph{MNS} receives the entire weight of all clients within its range and this weight is equally distributed.

\begin{lemma}
    \label{lemma:minimum-neighborhood-set}
    For a facility placement profile $\mathbf{s}$,
    in a client equilibrium $\sigma$, each facility $f_j \in M$ of a minimum neighborhood set $M$ has a load of exactly $\ell_j(\s,\sigma)=\rho_\mathbf{s}$.
\end{lemma}

\begin{proof}
    Let $M$ be a MNS and $\sigma$ be an arbitrary client equilibrium.
    Let $$T = \argmin_{f_j \in \mathcal{F}} {\left(\ell_j(\s,\sigma)\right)}$$ be the set of facility agents who share the lowest load in $\sigma$.
    Let $\ell_T$ be the load of each facility agents in $T$, hence for each $f_j \in T$ we have $\ell_j(\s,\sigma) = \ell_T$.
    Assume for the sake of contradiction that $\ell_T < \frac{w(A_\mathbf{s}(M))}{|M|}$.
    Since $M$ is a MNS,
    we have $$\frac{w(A_\mathbf{s}(T))}{|T|} \geq \frac{w(A_\mathbf{s}(M))}{|M|}\text.$$
    Thus,

    \begin{align*}
        \sum_{f_j\in T}{\ell_j(\s,\sigma)} = |T| \cdot \ell_T & < |T| \frac{w(A_\mathbf{s}(M))}{|M|}  \\
        & \leq |T| \frac{w(A_\mathbf{s}(T))}{|T|} = w(A_\mathbf{s}(T))\text.
    \end{align*}

    Hence, there is at least one client agent $v_i$ in a range of at least one facility agent $f_a \in T$, who does not put her complete weight on the facility agents in $T$.
    Therefore, there is a facility agent $f_b \notin T$, with $\ell_b(\s,\sigma) > \ell_T$ and $\sigma(\mathbf{s},v_i)_b > 0$.
    However, since $\ell_b(\s,\sigma) > \ell_T$, $v_i$ would prefer to move weight away from $f_b$ to $f_a$.
    Thus, we arrive at a contradiction and in all client equilibria we have for each facility agent $f_j \in \mathcal{F}$ that $\ell_j(\s,\sigma) \geq \frac{|A_\mathbf{s}(M)|}{|M|}$.

    The facility agents in $M$ only have access to the clients in $A_\mathbf{s}(M)$.
    Thus, if for any facility agent $f_c \in M$ the utility is $\ell_c(\s,\sigma) > \frac{w(A_\mathbf{s}(M))}{|M|}$, there must be another facility agent $f_d \in M$ where $\ell_d(\s,\sigma) < \frac{w(A_\mathbf{s}(M))}{|M|}$ holds.
    Since this is not possible, we get for each facility agent $f_j \in M$ that $\ell_j(\s,\sigma) = \frac{w(A_\mathbf{s}(M))}{|M|}$.
\end{proof}

\subsection{Facility Loads in Polynomial Time}
\label{sec:poly-alg-load-balancing}

We present a polynomial-time combinatorial algorithm to compute the loads of the facility agents in a client equilibrium for a given facility placement profile $\mathbf{s}$.
As each facility agent only has $n$ possible strategies, this implies that the best responses of facility agents are computable in polynomial time.

\Cref{alg:load-balancing-utilities} iteratively determines a MNS $M$, assigns to each facility in $M$ the MNR and removes the facilities and all client agents in their range from the instance.
See \Cref{fig:load-balancing-algorithm} for an example of a run of the algorithm.

\begin{algorithm}
    \caption{computeUtilities($H=(V,E,w), \mathcal{F}$, $\mathbf{s}$)}
    \label{alg:load-balancing-utilities}
    \lIf{$\mathcal{F} = \emptyset$}{\KwRet}
    $M \gets$ computeMNS($H, \mathcal{F}, \mathbf{s}$)\;
    \For{$f_j \in M$}{
        $\ell_j(\s,\sigma) \gets \frac{w(A_\mathbf{s}(M))}{|M|}$\;
    }
    $H' \gets (V,E,w')$ \textbf{with} $ w'(v_i)=0$ \textbf{if} $v_i \in A_\mathbf{s}(M)$ \textbf{else} $w'(v_i)=w(v_i)$\;
    computeUtilities($H', \mathcal{F} \setminus M, \mathbf{s}$)\;
\end{algorithm}

\begin{figure}[h]
    \centering
    \begin{tikzpicture}[scale=0.8, transform shape]
        \begin{scope}
            [
            every node/.style = {circle, thick, draw, inner sep = 0 pt, minimum size = 15 pt}
            ]
            \node (1) [label=above left:$f_1$] at (0, 0) {};
            \node (2) at (1*\trheight, 0) {};
            \node (3) at (2*\trheight, 0) {};
            \node (4) at (3*\trheight, 0) {};
            \node (5) at (4*\trheight, 0) {};
            \node (6) [label=above right:$f_4$] at (5*\trheight, 0) {};
            \node (7) [label=right:$f_2$] at (2.5*\trheight, 0.95*\trheight) {};
            \node (8) [label=right:$f_3$] at (2.5*\trheight, -0.95*\trheight) {};
            \node (9) at (5*\trheight, -0.95*\trheight) {};
            \node (10) at (5*\trheight, 0.95*\trheight) {};
        \end{scope}
        \begin{scope}
            [
            every node/.style = {circle,fill,inner sep=2pt}
            ]
            \node (1a) at (1) {};
            \node (6a) at (6) {};
            \node (7a) at (7) {};
            \node (8a) at (8) {};
        \end{scope}
        \coordinate(box1a) at (1.4*\trheight,1.3*\trheight);
        \coordinate(box1b) at (-0.7*\trheight,-1.3*\trheight);
        \coordinate(box2a) at (4.4*\trheight,1.3*\trheight);
        \coordinate(box2b) at (1.6*\trheight,-1.3*\trheight);
        \coordinate(box3a) at (5.9*\trheight,1.3*\trheight);
        \coordinate(box3b) at (4.6*\trheight,-1.3*\trheight);
        \begin{scope}
            [
            every node/.style = {inner sep = 0 pt, minimum size = 5 pt}
            ]
            \node [label=below left:$S_1$] at (box1a) {};
            \node [label=below left:$S_2$] at (box2a) {};
            \node [label=below left:$S_3$] at (box3a) {};
        \end{scope}
        \begin{scope}
            [
            every path/.style = {thick, {Latex[length=2mm]}-}
            ]
            \draw
            (1) edge (2)
            (6) edge (5) edge (9) edge (10)
            (7) edge (2) edge (3) edge (4) edge (5)
            (8) edge (2) edge (3) edge (4) edge (5)
            ;
        \end{scope}
        \begin{scope}
            [
            every path/.style = {thick, dotted}
            ]
            \draw (box1a) rectangle (box1b);
            \draw (box2a) rectangle (box2b);
            \draw (box3a) rectangle (box3b);
        \end{scope}
    \end{tikzpicture}
    \caption[Example of algorithm computeUtilities]{
        An instance of the load balancing \limtmpmodel{} with a facility placement profile marked by dots and $10$ clients with weight $1$ each. \Cref{alg:load-balancing-utilities}
        successively finds and removes the minimum neighborhood sets $S_1=\{f_1\}$, $S_2=\{f_2, f_3\}$ and $S_3=\{f_4\}$.}
    \label{fig:load-balancing-algorithm}
\end{figure}
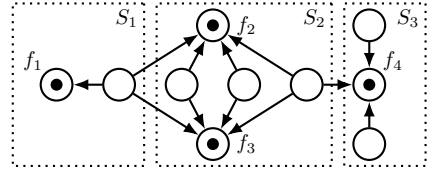

\noindent The key ingredient of \Cref{alg:load-balancing-utilities} is the computation of a MNS in
\Cref{alg:minimum-neighborhood-set}.
\begin{algorithm}
    \caption{computeMNS($H=(V,E,w), \mathcal{F}, \mathbf{s}$)}
    \label{alg:minimum-neighborhood-set}
    construct directed graph $G=(V', E_{st} \cup E_\text{Range})$\;
    $V' \gets \{s, t\} \cup V \cup \mathcal{F}$\;
    $E_{st} \gets \{(s, v_i, w(v_i)) \mid v_i \in V\} \cup \{(f_j, t, 0) \mid f_j \in \mathcal{F}\}$\;
    $E_\text{Range} \gets \{(v_i, f_j, w(v_i)) \mid v_i \in V, f_j \in A_\mathbf{s}(v_i)\}$\;
    possibleUtilities $\gets$ sorted$(\{x/y \mid x, y \in \mathcal{N}, 0 \leq x \leq w_\mathbf{s}(\mathcal{F}), 1 \leq y \leq k\})$\;
    \For{binary search over $i \in$ possibleUtilities}{
        $\forall{f_j \in \mathcal{F}}:$ capacity$((f_j, t)) \gets i$\;
        $h \gets$ maximum $s$-$t$-flow in $G$\;
        \leIf{$\text{value}(h) = i\cdot k$}{$i$ too small}{$i$ too large}
    }
    $T \gets \emptyset, \rho \gets$ largest $i \in$ possibleUtilities not too large\;
    \For{$f_j \in \mathcal{F}$}
    {
        $\forall{f_p \in \mathcal{F}}:$ capacity$((f_p, t)) \gets \rho$\;
        capacity$((f_j, t)) \gets \infty$\;
        start with flow from binary search for $i=\rho$\;
        \If{$\nexists$ augmenting path in $G$}{
            $T \gets T \cup \{f_j\}$\;
        }
    }
    \KwRet{T}\;
\end{algorithm}
Here, we first identify the MNR by a reduction to a maximum flow problem.
To this end, we construct a graph, where from a common source vertex $s$ demand flows through the clients to the facility agents in their respective ranges and then to a common sink $t$. See \Cref{fig:smallest-neighborhood-flow-graph} for an example of such a reduction.
By using binary search, we find the highest capacity value of the edges from the facility agents to the sink such that the flow can fully utilize all these edges.
This capacity value is the value of the MNR~$\rho_\mathbf{s}$. Note that by \Cref{lemma:limited-values-load-balancing} the MNR can only attain a limited number of values.
After determining the MNR, we identify the facility agents belonging to a MNS $M$ by individually increasing the capacity of the edge to the sink $t$ for each facility agent. Only if this does not increase the maximum flow, a facility agent belongs to $M$.
By reusing the flow for $\rho_\mathbf{s}$ a search for an augmenting path with the increased capacity is sufficient to determine if the flow is increased.

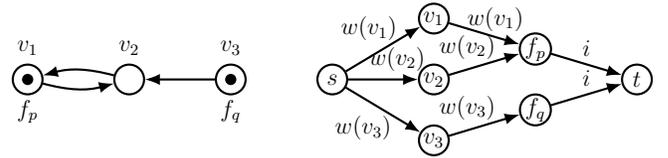
\begin{figure}[h]
    \Large
    \centering
    \begin{tikzpicture}[scale= 0.75, transform shape]
        \begin{scope}
            [
            every node/.style = {circle, thick, draw, inner sep = 0 pt, minimum size = 15 pt}
            ]
            \node (v1) [label=above:$v_1$] [label=below:$f_p$] at (-4.5*\trlength, 0) {};
            \node (v2) [label=above:$v_2$] at (-3*\trlength, 0) {};
            \node (v3) [label=above:$v_3$] [label=below:$f_q$] at (-1.5*\trlength, 0) {};

            \node (s) at (0, 0) {$s$};
            \node (v1a) at (1.5*\trlength, 0.9*\trlength) {$v_1$};
            \node (v2a) at (1.5*\trlength, 0) {$v_2$};
            \node (v3a) at (1.5*\trlength, -0.9*\trlength) {$v_3$};
            \node (p) at (3*\trlength, 0.45*\trlength) {$f_p$};
            \node (q) at (3*\trlength, -0.45*\trlength) {$f_q$};
            \node (t) at (4.5*\trlength, 0) {$t$};
        \end{scope}
        \begin{scope}
            [
            every node/.style = {circle,fill,inner sep=2pt}
            ]
            \node (1a) at (v1) {};
            \node (3a) at (v3) {};
        \end{scope}
        \begin{scope}
            [
            every node/.style = {above, midway, fill=none, draw=none},
            every path/.style = {thick, -{Latex[length=2mm]}}
            ]
            \draw
            (s) edge node[xshift=-.25cm] {$w(v_1)$} (v1a)
            (s) edge node[xshift=.25cm] {$w(v_2)$} (v2a)
            (s) edge node[below, xshift=-.35cm, yshift=.05cm] {$w(v_3)$} (v3a)
            (v1a) edge node[xshift=.2cm, yshift=-.1cm] {$w(v_1)$} (p)
            (v2a) edge node[xshift=-.3cm, yshift=-.05cm] {$w(v_2)$} (p)
            (v3a) edge node[xshift=-.3cm, yshift=-.05cm] {$w(v_3)$} (q)
            (p) edge node {$i$} (t)
            (q) edge node {$i$} (t)
            (v2) edge[bend right=15] (v1)
            (v1) edge[bend right=15] (v2)
            (v3) edge (v2)
            ;
        \end{scope}
    \end{tikzpicture}
    \caption[Example of construction for computing nodesWithSmallestNeighborhoodPerFacility]{Left: An instance of the load balancing \limtmpmodel{} with the graph~$H$ and the facility placement profile $\mathbf{s}$ marked by dots. Right: The maximum flow instance constructed by \Cref{alg:minimum-neighborhood-set}.
    }
    \label{fig:smallest-neighborhood-flow-graph}
\end{figure}

\noindent We first prove the correctness of \Cref{alg:minimum-neighborhood-set}:

\begin{theorem}
    \label{lemma:smallest-neighborhood-algorithm-correctness}
    For an instance of the load balancing \limtmpmodel{}, a facility placement profile $\mathbf{s}$, \Cref{alg:minimum-neighborhood-set} computes a MNS.
\end{theorem}
\begin{proof}
    We show that the MNR $\rho$ computed by the algorithm is correct by proving that $\rho$ is a lower and upper bound for $\rho_\mathbf{s}$.

    We show that for each set of facility agents $T$, we get $\rho \leq \frac{w(A_\mathbf{s}(T))}{|T|}$.
    To this end, consider the maximum flow for $i=\rho$. The value of this flow must be $\text{value}(h)=k \rho$, since $\rho$ is below the threshold found by the binary search.
    As the total capacity of the edges leaving the source $s$ towards vertices $v_i \in A_\mathbf{s}(T)$ is upper bounded by $w(A_\mathbf{s}(T))$ and every vertex $f_p$ with $f_p \in T$ is only reachable via vertices $v_i \in A_\mathbf{s}(T)$, the total inflow to the vertices $f_p \in T$ is $w(A_\mathbf{s}(T))$. Furthermore, the capacity of each edge from a facility vertex to the sink vertex $t$ is exactly $\rho$, hence each of these edges carries a flow of exactly $\rho$. Thus, we get $|T| \rho \le w(A_\mathbf{s}(T))$ for every set of facility agents $T$.

    For the upper bound, we show that there is a set $T$ for which $\rho \geq \frac{w(A_\mathbf{s}(T))}{|T|}$.
    We consider the flow at $i=\rho+\delta$, the value immediately above $\rho$ in \emph{possibleUtilities}.
    We assume that for each set $T$, $\rho+\delta \leq \frac{w(A_\mathbf{s}(T))}{|T|}$.
    By \Cref{lemma:minimum-neighborhood-set}, there must be a weight distribution~$\sigma$, such that every facility agent receives $\rho+\delta$ load.
    Thus, setting the flow of every edge $(v_i,f_j)$ in $h$ to $\sigma(\mathbf{s},v_i)_j$ for each $v_i \in V, f_j \in \mathcal{F}$ results in a flow of $(\rho+\delta)k$.
    This leads to $\rho+\delta$ being below the threshold and, hence, we have a contradiction.
    Therefore, there must be a set of facility agents $T$ with $\rho+\delta > \frac{w(A_\mathbf{s}(T))}{|T|}$.
    By \Cref{lemma:limited-values-load-balancing}, there is no value between $\rho$ and $\rho+\delta$, which $\frac{w(A_\mathbf{s}(T))}{|T|}$ can attain. Thus, there must be a set~$T$ with $\rho \geq \frac{w(A_\mathbf{s}(T))}{|T|}$.

    It remains to show that the set of facility agents $M$ computed by the algorithm is indeed a MNS.
    By the feasibility of the total flow of $k\cdot \rho$ for the instance with capacity bounds of $\rho$, we have for every set of facility agents~$T$, $\frac{w(A_\mathbf{s}(T))}{|T|}\ge \rho$.
    For every $f_j \notin M$, there exists an augmenting path where the edge $(f_j,t)$ has capacity~$\infty$.
    Hence, there is a total flow strictly larger than $k\cdot \rho$ with flow of exactly $\rho$ through all $f_q \ne f_j$.
    As the flow through each $f_i$ is bounded by $w(A_\mathbf{s}(f_i))$, for every $T$ with $f_j \in T$, $\frac{w(A_\mathbf{s}(T))}{|T|} > \rho$. Therefore, $f_j$ does not belong to the MNS.

    For every $f_j \in M$, the absence of an augmenting path certifies that the flow is constrained by capacity representing the clients' spending capacities. Hence, $\frac{w(A_\mathbf{s}(T))}{|T|} = \rho$ for every $T \subseteq M$.
\end{proof}

\noindent With that, we bound the runtime of \Cref{alg:minimum-neighborhood-set}.

\begin{lemma}
    \label{lemma:smallest-neighborhood-algorithm-runtime}
    \Cref{alg:minimum-neighborhood-set} runs in $\mathcal{O}(\log(w_\mathbf{s}(\mathcal{F})k)nk(n+k))$.
\end{lemma}

\begin{proof}
    Since $|\text{possibleUtilities}| \leq w_\mathbf{s}(\mathcal{F})k$, the binary search needs $\log{w_\mathbf{s}(\mathcal{F})k}$ steps.
    In each iteration, the dominant part is the computation of the flow, since all other operations are executable in constant time or are linear iterations through $G$.
    Therefore, the runtime of the binary search is the runtime of $\log{w_\mathbf{s}(\mathcal{F})k}$ flow computations in $G$.
    For the loop, we need $k$ breadth-first searches to determine the existence of augmenting paths.

    The graph $G$ we create has $|V'| = n+k+2$ vertices and at most $|E'| \leq n+k+nk$ edges.
    These values are not changed throughout the algorithm.
    Thus, by using Orlin's algorithm~\cite{orlin-algo}, to compute the maximum flow in $\mathcal{O}(nk(n+k))$, which dominates the complexity of the loop and its augmenting path searches.
    Therefore, the algorithm runs in $\mathcal{O}(\log(w_\mathbf{s}(\mathcal{F})k)nk(n+k))$.
\end{proof}

\noindent We return to \Cref{alg:load-balancing-utilities} and prove correctness and runtime:

\begin{theorem}
    Given a facility placement profile $\mathbf{s}$, \Cref{alg:load-balancing-utilities} computes the agent loads for an instance of the load balancing \limtmpmodel{} in $\mathcal{O}(\log(w_\mathbf{s}(\mathcal{F})k)nk^2(n+k))$.
\end{theorem}
\begin{proof}

    \emph{Correctness:} By \Cref{lemma:minimum-neighborhood-set} the utilities determined for the client agents in the MNS $M$ are correct for the given instance.
    Also by \Cref{lemma:minimum-neighborhood-set}, the client equilibria of $\mathcal{F} \setminus M$ are independent of the facility agents in $M$ and the clients in $A_\mathbf{s}(M)$.
    Therefore, we can remove $M$, set the weight of each client $v_i \in A_\mathbf{s}(M)$ to $w(v_i)=0$ and proceed recursively.

    \emph{Runtime:}
    The recursive function is called at most $k$ times because the instance size is decreased by at least one facility agent in each iteration.
    Apart from the call to \Cref{alg:minimum-neighborhood-set}, all computations can be done in constant or linear time.
    Therefore, the algorithm runs in $\mathcal{O}(\log(w_\mathbf{s}(\mathcal{F})k)nk^2(n+k))$.
\end{proof}
\noindent \Cref{alg:minimum-neighborhood-set} implicitly computes a client equilibrium.

\begin{corollary}
    \label{cor:construct-client-eq}
    A client equilibrium can be constructed by using the flow values on the edges between a client and the facility agents of the MNSs computed during the binary search in \Cref{alg:minimum-neighborhood-set} as the corresponding client weight distribution.
\end{corollary}

\begin{proof}
    Let $\s$ be a facility placement profile and for each facility $f_j$ let $h_j$ be the maximum $s$-$t$-flow found by the binary search during the run of \Cref{alg:minimum-neighborhood-set}, which finds $f_j$ to be part of a MNS.
    We construct a client weight distribution $\sigma$ in the following way:
    For each pair $v_i, f_j$, we set $\sigma(\s, v_i)_j = h_j(v_i, f_j)$, i.e., equal to the flow between $v_i$ and $f_j$ in $h_j$.

    We now show that $\sigma$ is indeed a client equilibrium:
    Let~$v_i$ be an arbitrary client.
    The algorithm removes her from the instance (i.e., sets her weight to 0) in the first round of \Cref{alg:load-balancing-utilities}, where she has any facility $f_p$ of the MNS $M$ found in that round in her shopping range.
    Thus, all facilities $f_j$ with $\sigma(\s, v_i)_j > 0$ are part of $M$.
    By the limit on the outgoing capacity of these facilities in the binary search in \Cref{alg:minimum-neighborhood-set}, all facilities in $M$ have equal load in $\sigma$.
    Since the MNR is nondecreasing throughout the run of the algorithm, all facilities which are part of an MNS found in a later iteration, have a equal or higher load in $\sigma$ than the facilities in $M$.
    Therefore, client $v_i$ cannot improve by moving her weight.
\end{proof}

\subsection{Existence of Subgame Perfect Equilibria}
\label{sec:load-eq}
We show that the load balancing \limtmpmodel{} always possesses SPE using a lexicographical potential function.
For that, we show that when a facility agent $f_p$ changes her strategy, no other facility agent $f_q$'s load decreases below $f_p$'s new load.

\begin{lemma}
    \label{lemma:no-lower-load-balancing}
    Let $\mathbf{s}$ be a facility placement profile and $f_p$ a facility agent with an improving move $s'_p$ such that $\ell_p((s'_p,s_{-p}),\sigma)>\ell_p(\s,\sigma)$, where $\sigma$ is a client equilibrium.
    For every facility agent $f_q$ with $\ell_q((s'_p,s_{-p}),\sigma) <\ell_q(\s,\sigma)$, we have that $$\ell_q((s'_p,s_{-p}),\sigma) \ge \ell_p((s'_p,s_{-p}),\sigma).$$
\end{lemma}

\begin{proof}
    Let $Q$ be the set of facility agents $f_q$ with $\ell_q((s'_p,s_{-p}),\sigma) <\ell_q(\s,\sigma)$.
    Let $$Q_\text{min} = \argmin_{q \in Q}\{\ell_q((s'_p,s_{-p}),\sigma)\}.$$
    Now, we distinguish two cases for $f_p$:

    \emph{Case 1: } $f_p \in Q_\text{min}$.
    The statement is trivially true.

    \emph{Case 2: } $f_p \notin Q_\text{min}$.
    All facility agents in $Q_\text{min}$ have the same clients in their ranges as before.
    Thus, there must be a client $v_i$, who has decreased her weight on a facility agent $f_r \in Q_\text{min}$ and increased her weight on a facility agent $f_s \notin Q_\text{min}$.
    Hence, we have $\ell_s((s'_p,s_{-p}),\sigma) \leq \ell_r(\s,\sigma)$ as otherwise, the client $v_i$ would not put weight on $f_s$.
    We assume $f_p \not= f_s$.
    As $\sigma$ is a client equilibrium, we have that
    $\ell_r(\s,\sigma) \leq \ell_s(\s,\sigma)$.
    This implies $$\ell_s((s'_p,s_{-p}),\sigma) <\ell_s(\s,\sigma)$$ which contradicts $f_s \notin Q_\text{min}$.
    Therefore, $f_p = f_s$ and $$\ell_r((s'_p,s_{-p}),\sigma) \geq \ell_p((s'_p,s_{-p}),\sigma),$$ which means that for each facility agent $f_q \in Q$, it holds that $$\ell_q((s'_p,s_{-p}),\sigma) \geq \ell_p((s'_p,s_{-p}),\sigma).$$
\end{proof}

\noindent With this lemma, we prove the FIP and, hence, existence of a SPE by a lexicographic potential function argument.

\begin{theorem}
    The load balancing \limtmpmodel{} has the FIP.

\end{theorem}
\begin{proof}
    Let $\Phi(\s) \in \mathbb{R}^k$ be the vector that lists the loads $
    \{\ell_1(\s,\sigma),$ $ \ell_2(\s,\sigma), \dots, \ell_k(\s,\sigma)\}$ in an increasing order.

    Let $\mathbf{s}$ be a facility placement profile and $f_p$ a facility agent with an improving move $s'_p$ such that $\ell_p((s'_p,s_{-p}),\sigma)>\ell_p(\s,\sigma)$, where $\sigma$ is a client equilibrium.
    We show that $\Phi(s'_p, s_{-i}) <_\text{lex} \Phi(\s)$.
    Let $\Phi(\s)$ be of the form
    $
    \Phi(\s)=(\phi_1,\ldots,\phi_{\alpha}, \ell_p(\s,\sigma), \phi_{\alpha+1},\ldots,\phi_{\beta},$ $\phi_{\beta+1},\ldots,\phi_{k-1})\text,
    $
    for some $\alpha\le \beta \le k-1$, such that for every $1\le j \le \beta$, $$\phi_j < \ell_p((s'_p,s_{-p}),\sigma)$$ and for every $j \ge \beta+1$, $$\phi_j \ge \ell_p((s'_p,s_{-p}),\sigma).$$

    By \Cref{lemma:no-lower-load-balancing}, we have for all facility agents $f_q$ with a load $\ell_q(\s,\sigma) \in \{\phi_1,\ldots,\phi_\beta\}$ that their loads do not decrease.
    and for agents $f_q$ with $\ell_q(\s,\sigma) \in \{\phi_{\beta+1},\ldots,\phi_k\}$ we have $\ell_q((s'_p,s_{-p}),\sigma) \ge \ell_p((s'_p,s_{-p}),\sigma)$.
    With the improvement of $f_p$, $
    \Phi(s'_p,s_{-p}) >_{\text{lex}} \Phi(\s)
    $ holds.
    By \Cref{lemma:limited-values-load-balancing}, there is a finite set of values that the loads can attain, thus, $\Phi$ is an ordinal potential function and the game has the FIP.
\end{proof}

\section{Comparison with Utility Systems}
\label{sec:utility-systems}

A \emph{utility system} (US)~\cite{vetta-utility-system} is a game, in which agents gain utility by selecting a set of actions, which they choose from a collection of subsets of a groundset available to them. Utility is assigned to the agents by a function of the set of selected actions of all agents.

\begin{definition}[Utility Systems (US) ~\cite{vetta-utility-system}]
    A utility systems consists of a set of $k$ agents, a groundset $V_p$ for each agent $p$, a strategy set of feasible action sets $\mathcal{A}_p \subseteq 2^{V_p}$, a social welfare function $\gamma : 2^{V^*} \to \mathbb{R}$ and a utility function $\alpha_p : 2^{V^*} \to \mathbb{R}$ for each agent~$p$, where
    $V^*=\cup_{p \in P}{V_p}$.

    For a strategy vector $(a_1, \dots, a_k)$, let $A=a_1 \cup \dots \cup a_k$ and
    $A \oplus a_p'$ denotes the set of actions obtained if agent $p$ changes her action set from $a_p$ to $a_p'$.
    A game is a utility system if
    $\alpha_p(A) \geq \gamma(A \oplus \emptyset)\text.$
    The utility system is \emph{basic} if $\alpha_p(A) = \gamma(A \oplus \emptyset)$ and is \emph{valid} if
    $\sum_{p \in P}{\alpha_p(A)} \leq \gamma(a)\text.$
\end{definition}

\noindent We show that the load balancing \limtmpmodel{} is not a basic but a valid US and we can apply the corresponding bounds for the PoA but not the existence of stable states.
\begin{lemma}
    The load balancing \limtmpmodel{} is a utility system.
\end{lemma}

\begin{proof}
    Each facility agent $f_p$ corresponds to a agent $p$ in the US with the groundset $$V_p = \{v^p \mid \text{ for each } v \in V\}$$ and
    the action set $$\mathcal{A}_p = \{\{v^p \} \cup \{w^p \in V \mid (v,w) \in E \} \mid v \in V\}.$$
    We define $$\gamma(X) = \sum_{v \in V \mid \exists p : v^p \in X}{w(v)}\text,$$ which corresponds to
    the sum of weights of covered clients and $\alpha_p(A)$ to correspond to the load of $f_p$ which can be expressed as a function of the sets of clients in range of each facility.
    To show the US condition $\alpha_p(A) \geq \gamma(A \oplus \emptyset)$, we
    let the social welfare decrease by a value of $x$ through a removal of agent~$p$ from strategy profile $a$ resulting in a new strategy profile $a_{-p}$.
    Hence, clients with a total weight of $x$ were only covered by $p$ in $a$.
    Thus, agent $p$ must receive at least $x$ utility in $a$, and the condition is fulfilled.
\end{proof}

\noindent As $\gamma$ merely depends on the covered clients, we have for every $X, Y \subset V^*$ with $X \subseteq Y$ and any $v^p \in V^*\setminus Y$, we have that $\gamma(X \cup \{v^p\}) -\gamma(X) \ge \gamma(X \cup \{v^p\}) -\gamma(X)
$. Hence, the following lemma is immediate.
\begin{lemma}
    The function $\gamma$ is non-decreasing and submodular.
\end{lemma}

\noindent We now show that the load balancing \limtmpmodel{} is a valid but not basic US.

\begin{theorem}
    The load balancing \limtmpmodel{} is a valid, but not a basic US.
\end{theorem}

\begin{proof}

    The following example proves that the load balancing \limtmpmodel{} is not a basic US. Let $H = (V,E,w)$ with $V = (v_1,v_2)$, $w(v_1) = w(v_2) = 1$ and $E = \{(v_1,v_2),(v_2,v_1)\}$. Furthermore, we have two facility agents $f_p$ and $f_q$ with $\mathbf{s} = (v_1,v_2)$. Removing agent $f_p$ does not change the weighted participation rate $W(\mathbf{s})$ since all clients are still covered. However, the utility of the removed facility agent $f_p$ is equal to $1$. Hence, equality in the utility system condition does not hold and the US is not basic.

    To show that the load balancing \limtmpmodel{} is a valid US, note that each client~$v$ who is in the attraction range of at least one facility agent distributes her total weight $w(v)$ among the agents. All other clients are uncovered and hence, their distributed weight is equal to $0$.
    Thus, the total weight $\sum_{v_i \in C(\mathbf{s})}w(v_i)$ distributed by clients, which is equal to the sum of the facility agents' loads, is equal to the value of the welfare function $W(\mathbf{s})$.
\end{proof}

\noindent We are now able to apply the PoA bound of \cite{vetta-utility-system} to our model.

\begin{corollary}
    The PoA of the load balancing \limtmpmodel{} is at most~2.
\end{corollary}

\section{Arbitrary Client Behavior}

In the following, we investigate the quality of stable states of the \limtmpmodel{} with arbitrary client behavior, i.e., the client costs are arbitrarily defined, and provide an upper and lower bound for the PoA as well as a lower bound for the PoS. Additionally, we show that computing the social optimum is NP-hard.

\begin{theorem}
    \label{thm:lar-poa-2}
    The PoA of the \limtmpmodel{} is at most $2$.
\end{theorem}

\begin{proof}
    Fix a \limtmpmodel{} with $k$ facility agents. Let OPT be a facility placement profile that maximizes social welfare and let $(\text{SPE},\sigma_{\text{SPE}})$ be a SPE.
    Let $C(\text{SPE})$ be the set of clients $v_i$ which are covered in SPE and $C(\text{OPT})$ be the set of clients~$v_i$ which are covered in OPT, respectively. Let UNCOV $= C(\text{OPT}) \setminus C(\text{SPE})$ be the set of clients which are covered in OPT but uncovered in SPE.

    Assume that $W(\text{OPT}) > 2 W(\text{SPE})$ and hence, $\sum_{v \in \text{UNCOV}}w(v) > W(\text{SPE})$.
    Then, there exists a facility agent $f_p$ that receives in OPT more than $\frac{W(\text{SPE})}{k}$ load from the clients in UNCOV.
    Now consider a facility agent $f_q$ with load $\ell_q \left( \text{SPE},\sigma_{\text{SPE}} \right) \leq \frac{W(\text{SPE})}{k} $. By changing her strategy and selecting the position of facility agent $f_p$ in OPT, agent $f_q$ receives the weight of all clients in UNCOV which are covered by $f_p$ in OPT since they are currently uncovered and, therefore, obtains more than $\frac{W(\text{SPE})}{k}$ load. As this contradicts the assumption of $\left( \text{SPE},\sigma_{\text{SPE}} \right)$ being a SPE, we have that $W(\text{OPT}) \le 2 W(\text{SPE})$.
\end{proof}
\noindent We contrast the upper bound of the PoA with a lower bound for the PoA and PoS.

\begin{theorem}
    The PoA and PoS of the \limtmpmodel{} is at least $2-\frac1k$.
\end{theorem}

\begin{proof}
    We prove the statement by providing an example of an instance~\emph{I} which has a unique equilibrium.
    Let~$x \geq 4$, $x \in \mathbb{N}$.
    We construct a \limtmpmodel{} with $k$ facility agents, a host graph $H(V,E,w)$ with $V = \{v_1, \ldots , v_k,$ $v_{1,1}, \ldots, v_{1,x-1},$ $ v_{2,1},$ $\ldots, v_{k-1,x-1}, v_{k,1}, \ldots, v_{k,k x}\}$, for all $v \in V$, $w(v) = 1$ and $E = \{(v_i, v_{i,j}) \mid i \in [1,k-1], j \in [1,x-1] \}\ \cup \{(v_k, v_{k,i}) \mid i \in [1,k x]\}\ \cup \{(v_{k,i}, v_{i,1}) \mid i \in [1,k-1]\}$.
    See \Cref{fig:poa}.
    \begin{figure}
        \centering
        \begin{tikzpicture}[scale= 0.8, transform shape]
            \begin{scope}
                [
                every node/.style = {circle, thick, draw, inner sep = 0 pt, minimum size = 12 pt}
                ]
                \node (c) [label=10:$v_{k}$]at (0, -0.25*\trlength) {};
                \node (1) [label={[label distance=-2]above:$v_{k,1}$}] at (-4*\trlength, -\trlength) {};
                \node (3) [label=left:$v_{k, k-1}$] at (0, -\trlength) {};
                \node (4) [label=left:$v_{k, k}$] at (1*\trlength, -\trlength) {};
                \node (6) [label={[label distance=-2]above:$v_{k, kx}$}] at (2.5*\trlength, -\trlength) {};

                \node (1a) [label=right:$v_{1,1}$]at (-4 *\trlength, -1.8*\trlength) {};
                \node (1b) [label=right:$v_1$]at (-4 *\trlength, -2.55*\trlength) {};
                \node (1c) [label={left:$v_{1,2}$}]at (-4 *\trlength-0.6*\trheight, -2.6*\trlength -0.6*\trheight) {};
                \node (1e) [label={right:$v_{1,x-1}$}]at (-4 *\trlength+0.6*\trheight, -2.6*\trlength -0.6*\trheight) {};

                \node (3a) [label=right:$v_{k-1,1}$]at (0, -1.8*\trlength) {};
                \node (3b) [label=right:$v_{k-1}$]at (0, -2.55*\trlength) {};
                \node (3c) [label={left:$v_{k-1,2}$}]at (0-0.6*\trheight, -2.6*\trlength -0.6*\trheight) {};
                \node (3e) [label={right:$v_{k-1,x-1}$}]at (0+0.6*\trheight, -2.6*\trlength -0.6*\trheight) {};

            \end{scope}
            \coordinate(box1a) at (-4.4*\trlength,0.15*\trlength);
            \coordinate(box1b) at (2.9*\trlength,-1.3*\trlength);
            \coordinate(box2a) at (-4.8*\trlength-0.6*\trheight,-1.5*\trlength);
            \coordinate(box2b) at (-2.9*\trlength+0.6*\trheight,-3.4*\trlength);
            \coordinate(box3a) at (-1.1*\trlength-0.6*\trheight,-1.5*\trlength);
            \coordinate(box3b) at (1.4*\trlength+0.6*\trheight,-3.4*\trlength);
            \begin{scope}
                [
                every node/.style = {inner sep = 0 pt, minimum size = 5 pt}
                ]
                \node (2) at (-2*\trlength, -\trlength) {\dots};
                \node (5) at (1.75*\trlength, -\trlength) {\dots};

                \node (1d) at (-4*\trlength, -2.6*\trlength -0.6*\trheight) {\dots};
                \node (3d) at (0, -2.6*\trlength -0.6*\trheight) {\dots};

                \node at (-2*\trlength, -1.9*\trlength) {\dots};
                \node at (-2*\trlength, -2.55*\trlength) {\dots};
                \node at (-2*\trlength, -2.6*\trlength -0.6*\trheight) {\dots};

                \node (b1) [label=below right:$S_k$] at (box1a) {};
                \node (b2) [label=below right:$S_1$] at (box2a) {};
                \node (b3) [label=below right:$S_{k-1}$] at (box3a) {};

            \end{scope}
            \begin{scope}
                [
                every path/.style = {thick, {Latex[length=2mm]}-}
                ]
                \draw
                (c) edge (1) edge (3) edge (4)  edge (6)
                (1) edge (1a)
                (1b) edge (1a) edge (1c)  edge (1e)
                (3) edge (3a)
                (3b) edge (3a) edge (3c)  edge (3e)
                ;
            \end{scope}
            \begin{scope}
                [
                every path/.style = {thick, dotted}
                ]
                \draw (box1a) rectangle (box1b);
                \draw (box2a) rectangle (box2b);
                \draw (box3a) rectangle (box3b);
            \end{scope}
        \end{tikzpicture}
        \caption{The host graph $H$ of an instance $I$ of the \limtmpmodel{} with arbitrary client behavior with a unique SPE.}
        \label{fig:poa}
    \end{figure}
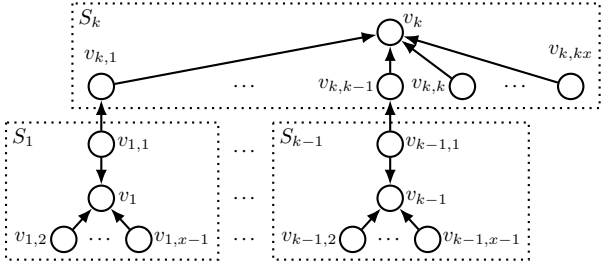

    We note that $H$ consists of a large star $S_k$ with central vertex $v_k$, leaf vertices $(v_{k,1}, \ldots, v_{k,k x})$ and $k-1$ small stars $S_i$ for $i \in [1,k-1]$ with central vertices $v_i$ and leaf vertices $(v_{i,1}, \ldots, v_{i,x-1})$. Each star $S_i$ is connected to $S_k$ via an edge between a leaf vertex of $S_k$ and $S_i$, i.e., $(v_{k,i}, v_{i,1})$.

    If the $k$ facility agents are placed on $\mathbf{s}_{\text{OPT}} = (v_1, \ldots, v_k)$, all clients are covered by exactly one facility. Hence, $W(\text{OPT}(H,k)) = |V| = kx + k +(k-1)(x-1)$.

    In any equilibrium, a facility $f_j$ for $j \in [1,k]$ must receive a load of at least $\frac{kx+1}{k} = x +\frac1k$ as otherwise switching to vertex $v_k$ with $kx+1$ adjacent vertices yields an improvement.
    However, any other vertex in $H$ has at most $x-1$ adjacent vertices, hence, every facility agent gets a load of at most~$x$. Therefore, the unique SPE is $\mathbf{s}_{\text{SPE}} = (v_k, \ldots, v_k)$ with $W(\s_\text{SPE}) = k x+1$ and
    $$\text{PoA} = \text{PoS} = \frac{kx + k +(k-1)(x-1)}{k x+1} = \frac{(2k-1)x+1}{kx + 1}.$$
    We get
    $\lim\limits_{x \rightarrow \infty} \left( \frac{(2k-1)x+1}{kx + 1}\right) = \frac{2k-1}{k} = 2 - \frac1k.$
\end{proof}

\noindent By a reduction from \textsc{3SAT}, we show that computing OPT$(H,k)$ is an NP-hard problem.

\begin{theorem}
    \label{thm:nphard}
    Given a host graph $H$ and a number of $k$ facilities, computing the facility placement maximizing the weighted participation rate OPT$(H,k)$, is NP-hard.
\end{theorem}

\begin{proof}
    We prove the theorem by giving a polynomial time reduction from the NP-hard \textsc{3SAT} problem.

    For a \textsc{3SAT} instance $\phi$ with a set of clauses $C$ and a set of variables $X$,
    we create a \limtmpmodel{} instance
    with $k = |X|$ facility agents where the host graph $H(V_X \cup V_C,E_X \cup E_C,w)$ is defined as follows:
    \begin{align*}
        w(v)&=\{1 \mid v \in V_X \cup V_C\} \\
        V_X&=\{v_x, v_{\neg x} \mid x \in X\}\\
        V_C&=\{v_c \mid c \in C\}\}\\
        E_X&=\{(v_x, v_{\neg x}), (v_{\neg x}, v_x) \mid x \in X\}\\
        E_C&=\{(v_c, v_l) \mid c \in C, \text{ literal } l \in c\},
    \end{align*}
    where $v_l=v_x$ if the contained variable $x$ is used as a true literal in~$c$, and $v_l=v_{\neg x}$, otherwise. See \Cref{fig:3sat} for an example.
    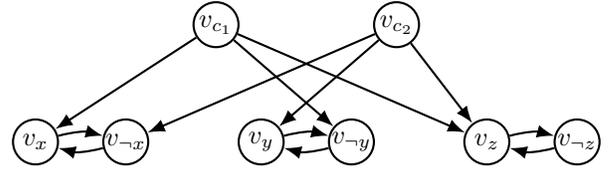
\begin{figure}
        \centering
        \begin{tikzpicture}
            \begin{scope}
                [
                every node/.style = {circle, thick, draw, inner sep = 0 pt, minimum size = 17 pt}
                ]
                \node (c) at (-\trlength, 1.3*\trlength) {$v_{c_1}$};
                \node (d) at (\trlength, 1.3*\trlength) {$v_{c_2}$};

                \node (x1) at (-3*\trlength, 0) {$v_x$};
                \node (x2) at (-2*\trlength, 0) {$v_{\neg x}$};
                \node (y1) at (-0.5*\trlength, 0) {$v_y$};
                \node (y2) at (0.5*\trlength, 0) {$v_{\neg y}$};
                \node (z1) at (2*\trlength, 0) {$v_z$};
                \node (z2) at (3*\trlength, 0) {$v_{\neg z}$};
            \end{scope}
            \begin{scope}
                [
                every path/.style = {thick, -Latex}
                ]
                \draw
                (x1) edge[bend left=15] (x2)
                (x2) edge[bend left=15] (x1)
                (y1) edge[bend left=15] (y2)
                (y2) edge[bend left=15] (y1)
                (z1) edge[bend left=15] (z2)
                (z2) edge[bend left=15] (z1)

                (c) edge (x1) edge (y2) edge (z1)
                (d) edge (x2) edge (y1) edge (z1)
                ;
            \end{scope}
        \end{tikzpicture}
        \caption{An example of a corresponding host graph $H$ to the \textsc{3SAT} instance ${(x\vee \neg y \vee z)} \wedge {(\neg x\vee y \vee z)}$.}
        \label{fig:3sat}
    \end{figure}

    Let $\phi$ be satisfiable and $\alpha$ be an assignment of the variables satisfying $\phi$. We set $\mathbf{s} = (s_1,\ldots, s_k)$ such that for $i \in [1,k]$, $x_i \in X$, $s_i = v_{x_i}$ if $x_i$ is true in $\alpha$ and $s_i = v_{\neg x_i}$ otherwise. By $E_X$, $ v_{x_i}$ and $v_{\neg x_i}$ are covered by a facility agent either located on $ v_{x_i}$ or $v_{\neg x_i}$. To show that each client $v_c \in V_C$ is covered as well, consider the corresponding clause $c = l_1 \vee l_2 \vee l_3$. Since $\phi$ is satisfied, at least one of the literals is true, which means that at least one of $v_{l_1}$, $v_{l_2}$ and $v_{l_3}$ must be occupied by a facility in $\mathbf{s}$. Thus, if $\phi$ is satisfied, we get a placement where all clients are covered, which is optimal.

    Let $\mathbf{s}$ be a facility placement profile where all clients are covered. Note that this implies that for each $x \in X$ either $ v_{x}$ or $v_{\neg x}$ is occupied by a facility agent. Hence, all facilities are placed on vertices in $V_X$. We construct an assignment of the variables $\alpha$ as follows: $x = $ true, if $v_{x} \in \mathbf{s}$ and $x = $ false, if $v_{\neg x} \in \mathbf{s}$. Let $c \in C$ be an arbitrary clause in $\phi$. The corresponding vertices $v_c$ is covered by a facility agent which is placed on an adjacent vertex, $v_{l_1}$, $v_{l_2}$, or $v_{l_3}$. This implies that at least one of the literals $l_1$, $l_2$, and $l_3$ is true in $\alpha$ and therefore $c$ is satisfied. Hence, $\phi$ is satisfiable.
\end{proof}

\section{Conclusion and Future Work}
We provide a general model for non-cooperative facility location with both strategic facilities and clients. Our load balancing \limtmpmodel{} is a proof-of-concept that even in this more intricate setting it is possible to efficiently compute and check client equilibria. Also, in contrast to classical one-sided models and in contrast to Kohlberg's two-sided model, the load-balancing \limtmpmodel{} has the favorable property that stable states always exist and that they can be found via improving response dynamics. Moreover, our bounds on the PoA and the PoS show that the broad class of 2-FLGs is very well-behaved since the societal impact of selfishness is limited.

The load balancing \limtmpmodel{} is only one possible realistic instance of a competitive facility location model with strategic clients; other objective functions are conceivable, e.g., depending on the distance and the load of all facilities in their shopping range. Also, besides the weighted participation rate other natural choices for the social welfare function are possible, e.g., the \emph{total facility variety} of the clients, i.e., for each client, we count the facilities in her shopping range. This measures how many shopping options the clients have. Moreover, we are not aware that the total facility variety has been considered for any other competitive facility location model.

\printbibliography

\end{document}